\documentclass{nws}
\usepackage{natbib}
\usepackage{placeins}
\usepackage[colorlinks=true, linkcolor=blue, urlcolor=blue, citecolor=blue]{hyperref}
\usepackage{booktabs}           
\usepackage{graphics}          
\usepackage{graphicx}          
\usepackage{algpseudocode}
\usepackage{enumitem}

\newtheorem{theorem}{Theorem}[section]

\begin{document}

\title[Network Science]{Spreaders in the Network SIR Model:\\ An Empirical Study}

\author[Macdonald et al.]{Brian Macdonald$^{1,2,5}$, Paulo Shakarian$^{1,3,5}$, Nicholas Howard$^{1,2}$, \& Geoffrey Moores$^{1,3,4}$\\
$^1$Network Science Center, United States Military Academy, West Point, NY 10996\\
$^2$Department of Mathematical Sciences, United States Military Academy, West Point, NY 10996\\
$^3$Department of Electrical Engineering and Computer Science, United States Military Academy, West Point, NY 10996\\
$^4$Department of Physics, United States Military Academy, West Point, NY 10996\\
$^5$These authors contibuted equally to this work.\\
\email{brian.macdonald@usma.edu,paulo@shakarian.net,nicholas.howard@usma.edu,geoffrey.moores@usma.edu}
}

\date{\today}

\maketitle

\begin{abstract}
We use the susceptible-infected-recovered (SIR) model for disease spread over a network, and empirically study how well various centrality measures perform at identifying which nodes in a network will be the best spreaders of disease on $10$ real-world networks. We find that the relative performance of degree, shell number and other centrality measures can be sensitive to $\beta$, the probability that an infected node will transmit the disease to a susceptible node.  We also find that eigenvector centrality performs very well in general for values of $\beta$ above the epidemic threshold.
\end{abstract}

\section{Introduction}
		The susceptible-infected-recovered (SIR) model, first introduced in \cite{anderson79} is a popular model for disease spread.  In recent years, this model has been applied to social networks - situations where the interactions of individuals are modeled as a graph.  A key problem relating to this model when considering a network structure is how to identify the nodes that, if initially infected, will result in the greatest portion of the population (in expectation) also becoming infected.  These nodes are often referred to as ``spreaders.''  Unfortunately, a modification of the proof of a related problem in \cite{chen10} shows that exactly computing the expected number of infected individuals in a networked-structured population given a single initial infectee is $\#P$-hard.  This implies that solving this problem exactly is likely beyond the ability of today's computer systems.  However, the literature on complex networks has provided various centrality measures that can be used as heuristics.  So, inspired by the work of  
		\cite{InfluentialSpreaders_2010}, which empirically examines the use of degree, betweenness, and shell number for identifying spreaders, we conduct a comprehensive evaluation of $10$ different centrality measures on $10$ real-world social network data-sets from various domains (e-mail, disease spread, blogging, power, autonomous system, and collaboration).  The major contributions of our work are two-fold.  First, we show that the ability of a centrality measure to identify spreaders in the SIR model can be sensitive to the $\beta$ parameter, the probability of infection.  Second, we find that, in general, eigenvector centrality performs very well for values of $\beta$ above the epidemic threshold.
   
    With respect to our first major contribution, we carefully selected the $\beta$ parameter based on $\beta'$, the \textit{epidemic threshold} of the network.  We can be sure that a contagion can spread to a significant portion of the network for $\beta > \beta'$, and we studied a variety of different values for $\beta$ above this threshold.
    
    In Figure \ref{cond-mat-1} and \ref{cond-mat-2},  we give an example of a network where shell number outperforms degree for one value of $\beta$, but degree outperforms shell number for another value of $\beta.$  In Section \ref{results}, we give additional examples illustrating that the imprecision functions of other centrality measures, as well as the choice of the ``best'' centrality measure, can be sensitive to $\beta$ as well.

   \begin{figure}
        \begin{center}
            \includegraphics[width=0.8\linewidth]{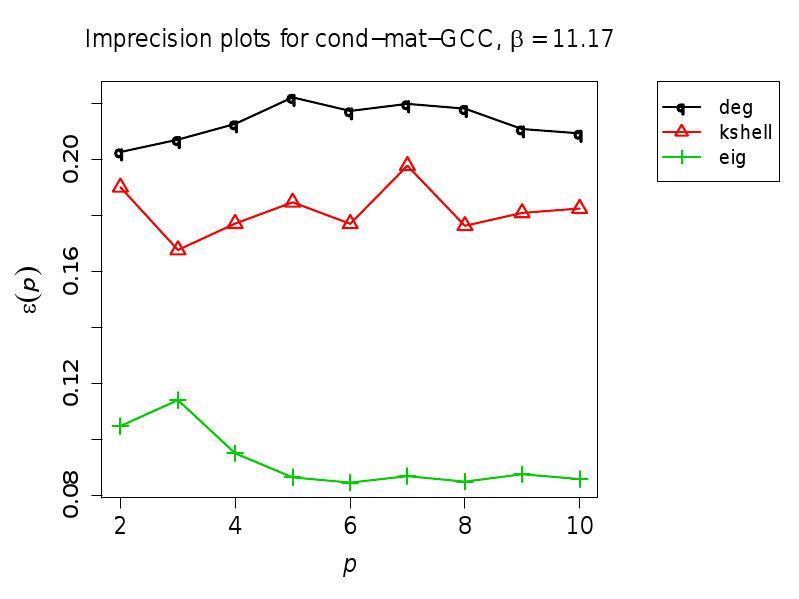}
        \end{center}
    \caption{ Imprecision versus $p$ 
    for the cond-mat network with $\beta=11.17$.  Notice that for this $\beta$, $k$-shell has a lower imprecision, meaning that $k$-shell outperforms degree. See Section \ref{defs} for the definitions of imprecision function and $p$.
}
        \label{cond-mat-1}
    \end{figure}
        \begin{figure}
            \begin{center}
                \includegraphics[width=0.8\linewidth]{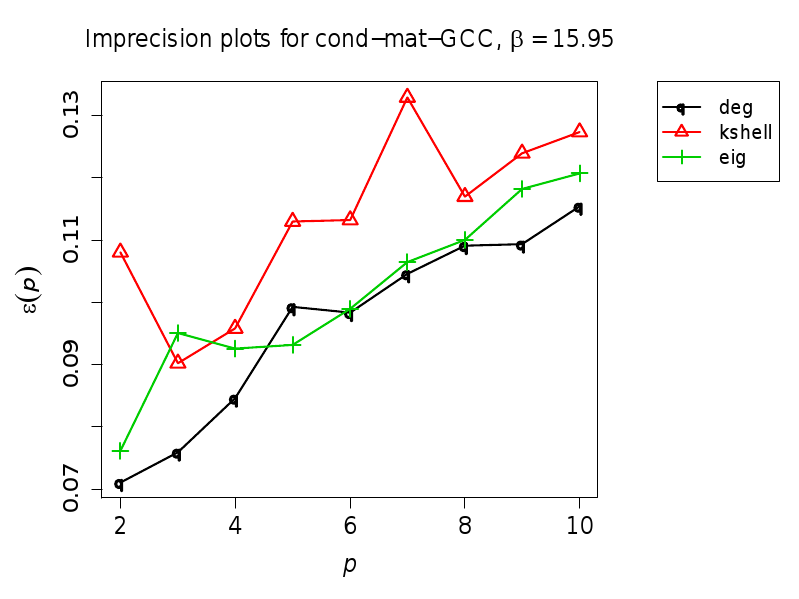}
            \end{center}
        \caption{Imprecision plots vs. $p$ for the cond-mat network with $\beta=15.95$.  Notice that for this $\beta$, degree has a lower imprecision, meaning that degree outperforms $k$-shell, the opposite of what we saw in Figure \ref{cond-mat-1}.}
            \label{cond-mat-2}
        \end{figure}

    As for our second major contribution, we found that eigenvector centrality consistently outperformed all other measures considered, including both shell number and degree (which were considered by Kitsak et al.), in all but one of the networks examined.  See Figure \ref{eig-kshell} for a comparison of $k$-shell (the best performing centrality measure of Kitsak et al.) with eigenvector centrality.  Also, if we average over all of our networks, including the one where eigenvector was not the best, we find that, on average, eigenvector centrality outperforms the other measures.

      \begin{figure}
                \begin{center}
                    \includegraphics[width=0.8\linewidth]{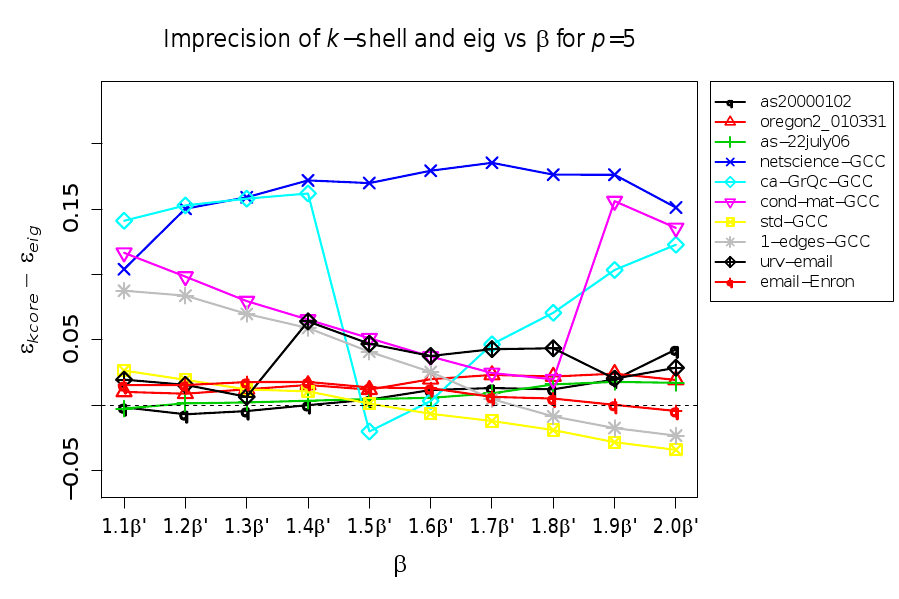}
                \end{center}
            \caption{Imprecision of $k$-shell minus the imprecision of eigenvector centrality.  Positive values indicate that $k$-shell has a higher imprecision than eigenvector centrality, which means that eigenvector centrality typically outperforms $k$-shell.}
                \label{eig-kshell}
            \end{figure}
    
The rest of this paper is organized as follows.  In Section \ref{sir}, we review the SIR model, discuss how the $\#P$-hardness proof of \cite{chen10} applies to this model, and describe how we calculate the epidemic threshold of a given complex network.  This is followed by a discussion of the various centrality measures we considered in Section~\ref{defs} along with a review of the description of the ``imprecision function''~\cite{InfluentialSpreaders_2010} used to measure the effectiveness of a centrality measure in identifying the top spreaders in a network.  We describe our experimental setup and datasets in Section~\ref{data} and give a description and discussion of the experimental results in Section~\ref{results}.

\section{The SIR Model}\label{sir}

As in~\cite{InfluentialSpreaders_2010}, we consider the classic susceptible-infected-recovered (SIR) model of disease spread introduced in \cite{anderson79}. In this model, all nodes in the network are in one of three states: susceptible (able to be infected), infected, or recovered (no longer able to infect or be infected). At each time step, any node infected in the last time step can infect any of its neighbors who are in a susceptible state with a probability $\beta$. After that time step, any node previously in an infected state moves into a recovered state and is no longer able to infect or be infected.

\subsection{Complexity}

In \cite{libai} and \cite{kleinberg}, the authors present a generalization of the SIR model known as the \textit{independent cascade} (IC) model.  In this model, the $\beta$ parameter can be different for each edge in the network.  They define the \textit{influence spread} of a set of nodes as the expected number of individuals in the population infected under the IC model given that the set was initially infected.  In \cite{chen10} this problem was shown to be $\#P$-hard.  Here we reconsider their proof, with some modification, to identify the influence spread of single node under the SIR model.

\begin{theorem}
\label{cmplx-thm}
Calculating the influence spread of a single node under the SIR models is $\#P$-hard.
\end{theorem}
\begin{proof}
We prove this theorem by showing a reduction from the known $\#P$-complete problem $s-t$ connectivity~\cite{valiant79}.  Let $G=(V,E)$ be a directed graph, where $V$ denotes the set of vertices, and $E$ denotes the set of edges.  Given two vertices $s,t \in V$, the goal is to determine the number of subgraphs of $G$ where $s$ is connected to $t$.  In \cite{chen10}, the authors point out that it is easy to see that this is equivalent to computing the probability that $s$ is connected to $t$ when each edge in $G$ has an independent probability of $0.5$ to be connected (and $0.5$ to be disconnected).  Hence, to embed the $s-t$ connectivity problem into the influence spread on the SIR model, we first calculate $M_s$, the expected number of infectees given initially infected node $s$ with $\beta=50$.  We then create $G'$ which is equivalent to $G$ but has an additional directed edge from $t$ to a new node $t'$.  Let $M_s'$ be the influence spread when we consider graph $G'$.  If $p(s,t,G)$ is the probability that $t$ is influenced by $s$ in $G$ (hence the solution to the $s-t$ connectivity problem) then $M_s' = M_s + p(s,t,G)\cdot\frac{\beta}{100}$.  Therefore, the solution to the $s-t$ connectivity problem can easily be obtained in polynomial time if we can efficiently find a solution to the influence spread problem under the SIR model.
\end{proof}

Theorem~\ref{cmplx-thm} tells us that exact methods for identifying the influence spread of individual nodes under the SIR models is likely not possible with today's computer systems.  Further, as $s-t$ connectivity has no known efficient approximation algorithm with a guarantee of accuracy, an approximation scheme for influence spread also seems unlikely.  Hence, much work on influence spread such as \cite{kleinberg} relies on estimating influence spread using simulation, which is often expensive computationally and even impractical for very large networks.  Therefore, in this paper, we look to evaluate various centrality measures from the literature as heuristics to identify spreaders under the SIR model.  We describe these centrality measures in Section~\ref{defs}.  Note that the centrality measures are not specifically designed to calculate influence spread under the SIR model, and they do not account for the infection probability $\beta$.  In the next section, we describe how we select the different $\beta$ parameters for the model in our experiments.

\subsection{Selecting the Infection Probability}

We note that for scale-free networks, having degree distribution $P(k) \sim k^{-\gamma}$, the literature shows that for $\gamma\leq 3$, the epidemic threshold of $\beta$ approaches $0$ as the number of nodes goes to infinity \cite{callaway00,cohen00}.  However, the networks we examine are of finite size and have various levels of ``scale-freeness'', based on the $R^2$ value of the linear correlation of a log-log plot of the degree distribution (see Section~\ref{data} for details). Instead, we explored $\beta$ values based on the epidemic threshold calculation in \cite{madar04}.  Using this method, the SIR model is mapped onto a bond percolation process.  Assuming a randomly connected network, the average number of influenced neighbors, $\langle n \rangle$ can be written
    \begin{equation}
    \langle n \rangle = \beta \cdot \sum\limits_k \frac{P(k)\cdot k \cdot (k-1)}{\langle k \rangle},
    \label{threshold}
    \end{equation}
where $k$ is the degree of a node, $P(k)$ is the probability of a node having degree $k$, and $\langle k \rangle$ is the average degree. Since an epidemic state can only be reached when $\langle n \rangle >1$, and from (\ref{threshold}) we have
\begin{equation}
\beta > \left( \sum\limits_k \frac{P(k)\cdot k \cdot (k-1)}{\langle k \rangle} \right)^{-1} = \beta'.
\end{equation}

We note that there is some work discussing the effect of different infection probabilities on spreading in \cite{InfluentialSpreaders_2010} and more recent and comprehensive study on the topic in \cite{castellano12}.  These works consider the effect of this parameter with respect to degree and shell decomposition (and betweenness in \cite{InfluentialSpreaders_2010}).  Here we consider these and many other centrality measures, and find that some of them, such as eigenvector centrality, outperform those in these previous works.

\section{Centrality Measures}\label{defs}
We now describe the centrality measures that we examine in our experiments.  We note that the major centrality measures in the literature can be classified as either radial (the quantity of certain paths originating from the node) or medial (the quantity of certain paths passing through the node) as done in Borgatti and Everett \cite{borgatti06}.  Based on the negative result concerning betweenness of \cite{InfluentialSpreaders_2010} and the intuitive association between high-radial nodes and spreading, we focused our efforts on radial measures.  While the work of \cite{InfluentialSpreaders_2010} compares shell number to degree and betweenness, we consider several other well-known radial measures in addition to degree, including closeness and eigenvector centrality.  As done in \cite{InfluentialSpreaders_2010}, we also develop ``imprecision functions'' for these centrality measures.

\subsection{Degree Centrality}
Of all the measures that we are examining, degree is perhaps the most simplistic measure - simply the total of incident edges for a given node.  As noted throughout the literature, such as \cite{wasserman1994social}, it is perhaps the easiest centrality measure to compute.  Further, in other diffusion processes, such as the voter model on undirected networks in~\cite{antal06}, it has been shown to be proportional to the expected number of individuals becoming infected\footnote{Technically, the work of \cite{antal06} proves that the fixation probability for a single mutant invader is proportional to the degree of that node.  However, the expected number of mutants, in the limit as time goes to infinity, can simply be computed by multiplying fixation probability by the number of nodes in thee network.}.  As pointed out in \cite{borgatti06}, degree is a radial measure as it is the number of paths starting from a node of length $1$.  Degree is one of three measures considered in \cite{InfluentialSpreaders_2010}.

\subsection{Shell Number}

The other radial measure considered in \cite{InfluentialSpreaders_2010}, shell number, or ``$k$-shell number'', is determined using shell decomposition \cite{Seidman83}.  High shell-number nodes in the network are often referred to as the ``core'' and are regarded by \cite{InfluentialSpreaders_2010} as influential spreaders under the SIR model.  Our results described later in the paper confirm this finding, although we also show that $k$-shell number was generally outperformed by eigenvector centrality. 
There have also been some more practical applications of this technique to find key nodes in a network.  For instance, \cite{borge12a, borge12} uses shell-decomposition to find individuals likely to initiate information cascades in an online social network while \cite{ShaiCarmi07032007} uses it to identify key nodes in a subset of autonomous systems on the Internet.

 An example of this process is shown in Figure~\ref{kshelldecomp}.  Given graph $G=(V,E)$, shell decomposition partitions a graph into shells and is described in the algorithm below.\\
\begin{algorithmic}
\State  Let $k_i$ be the degree of node $i$.  Set $S=1$.  Let $V_S$ denote the first shell of $G$.
\While{$|V|>0$} 
    
        \While{There exists $i$ such that $k_i=S$}
            \State Remove all $i \in V$ where $k_i=S$; 
            \State Also, remove all corresponding adjacent edges. 
            \State Place removed nodes into shell $V_S$.
        \EndWhile 
     \State $S$++
\EndWhile
\end{algorithmic}

\begin{figure} 
\begin{center}
\includegraphics[width=0.8\linewidth]{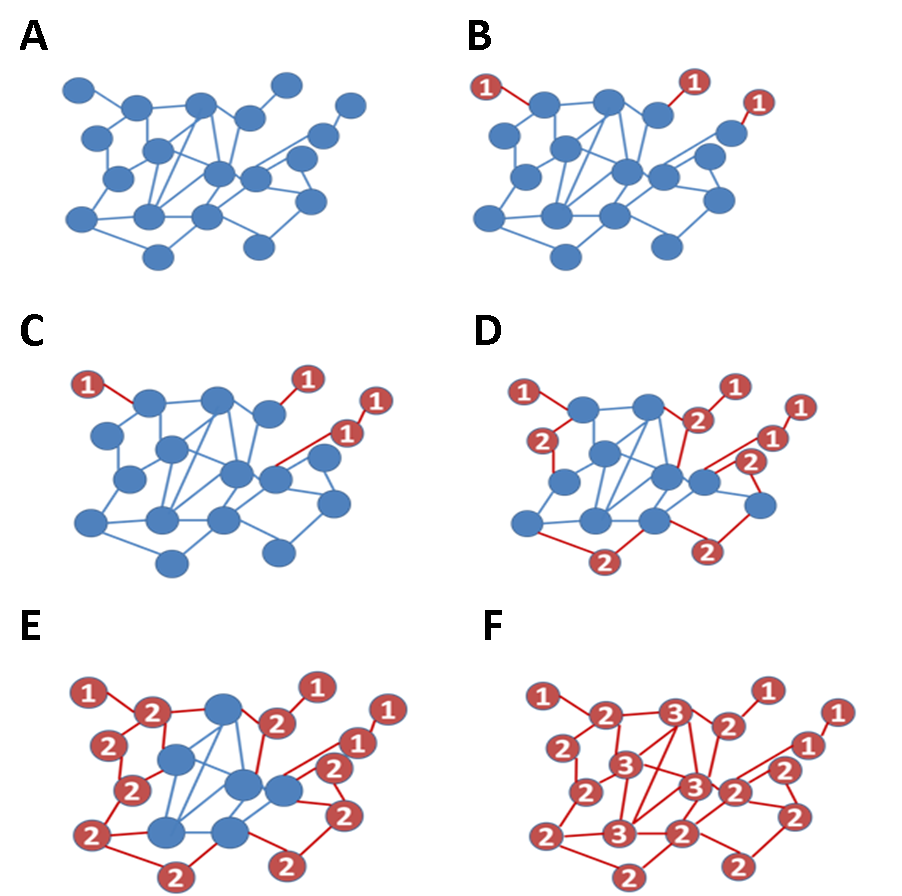}\\
\caption{\textit{Consider the progression of the graph above, where the elimination of nodes with degree 1 occurs in B and C.  D represents the first iteration for the second shell, and E represents the complete second shell (as well as the first).  F finalizes the decomposition with the third shell.}}
\label{kshelldecomp}
\end{center}
\end{figure}

\subsection{Betweenness Centrality}

The intuition behind high betweenness centrality nodes is that they function as ``bottlenecks'' as many paths in the network pass through them.  Hence, betweenness is a medial centrality measure.  Let $\sigma_{st}$ be the number of shortest paths between nodes $s$ and $t$ and $\sigma_{st}(v)$ be the number of shortest paths between $s$ and $t$ containing node $v$.  In \cite{freeman77}, betweenness centrality for node $v$ is defined as $\sum_{s\neq v \neq t}\frac{\sigma_{st}(v)}{\sigma_{st}}$.  In most implementations, including the ones used in this paper, the algorithm of \cite{brandes01} is used to calculate betweenness centrality.

\subsection{Closeness Centrality}

Another common measure from the literature that we examined is closeness ~\cite{freeman79cent}.  Given node $i$, its closeness $C_c(i)$ is the inverse of the average shortest path length from node $i$ to all other nodes in the graph. Intuitively, closeness measures how ``close'' it is to all other nodes in a graph.

	  Formally, if we define the shortest path between nodes $i$ to $j$ as function $d_G(i,j)$, we can express the average path length from $i$ to all other nodes as 
\begin{equation}
L_i = \frac{\sum_{j\in V\setminus i} d_G(i,j)}{|V|-1}.
\end{equation}

Hence, the closeness of a node can be formally written as 
\begin{equation}
C_c  (i) = \frac{1}{L_i} = \frac{|V|-1}{\sum_{j\in V\setminus i} d_G(i,j)}.
\end{equation}

\subsection{Eigenvector Centrality}

	The use of the principle eigenvector of the adjacency matrix of a network was first proposed as a centrality measure in \cite{bona72}.  Hence, the intuition behind eigenvector centrality is that it measures the influence of a node based on the sum of the influences of its adjacent nodes.  Given a network $V=(G,E)$ with adjacency matrix $A = (a_{ij})$, where $a_{ij} = 1$ if an edge exists between nodes $i$ and $j$, the eigenvector centrality of node $i$ satisfies

\begin{equation}
x_i = \frac{1}{\lambda} \sum\limits_{j \in V}a_{ij}x_j,
\end{equation}
for some $\lambda.$  If we define $x$ to be the vector of $x_i$'s, this relationship can be expressed as
\begin{eqnarray}
x=\frac{1}{\lambda}Ax, \, \textup{ or  } \,\,
Ax = \lambda x, 
\end{eqnarray}
which is the familiar equation relating $A$ with its eigenvalues and eigenvector.  The eigenvector centralities for the network are the entries of the eigenvector corresponding to the largest real eigenvalue.

\subsection{PageRank}
PageRank, introduced in \cite{Page98}, is computed for each node based on the PageRank of its neighbors.  Where $E$ is the set of undirected edges, $R_v,d_v$ is the PageRank and degree of $v$, and $c$ is a normalization constant, we have the relationship 
    $$R_v=c \cdot \sum_{v' | (v,v')\in E}\frac{R_{v'}}{d_{v'}}.$$  
An initial value for rank is entered for each node and the relationship is then computed iteratively until convergence is reached.  Intuitively, PageRank can be thought of as the importance of a node based on the importance of its neighbors.

\subsection{Neighborhood}

	The next centrality measure we consider is the ``neighborhood.''  Given a natural number $q$, the $q$-neighborhood of vertex $i$ is the number of nodes in the network that are distance $q$ or closer from node $i$.  For example, for $q=0$, this metric is $1$ for every node.  For $q=1$, this metric is identical to degree centrality of node $i$, since it is the number of nodes within a distance $1$ of $i$.  For $q=2$, this metric counts the number of nodes within a distance $2$ of $i$, so it counts $i$'s neighbors along with its neighbors' neighbors.  In our work, we computed neighborhoods using $q=2, 3, 5, 10$, and denoted these measures by $nghd2$, $nghd3$, $nghd5$, and $nghd10$, respectively.  We note that the work of \cite{chen12} develops a centrality measure with a similar intuition to the neighborhood and show it preforms well in identifying influential spreaders.

\subsection{The Imprecision Functions}

	We now define the imprecision functions from \cite{InfluentialSpreaders_2010} that are used to measure the effectiveness of a centrality measure in identifying influential spreaders. We also extend their definition for all centrality measures explored in this paper.  Let $N$ denote the number of nodes, and let $p$ be a real number between 0 and 100. The $pN/100$ highest efficiency spreaders, $\Upsilon_{eff}(p)$, are chosen based on number of nodes infected $M_i$ per node.  Similarly, a set $\Upsilon_{k_s}(p)$ is defined as the $pN/100$ predicted most efficient spreaders, chosen with priority to highest $k_s$ valued nodes.  Let
\begin{eqnarray}
M_{eff}(p) &= \sum\limits_{i \in \Upsilon_{eff}(p)}\frac{M_i}{pN}, \mathit{ and }\\
M_{k_s}(p) &= \sum\limits_{i \in \Upsilon_{k_s}(p)}\frac{M_i}{pN}. 
\end{eqnarray}
The imprecision function of $k_s$, $\epsilon_{k_s} (p)$, is defined as 
\begin{equation}
\epsilon_{k_s} (p)  = 1- \frac{M_{k_s}(p)}{M_{eff}(p)}
\end{equation}
     
Similarly, $\epsilon_{eig} (p)$ and $\epsilon_{deg} (p)$ are defined as
\begin{eqnarray}
        \epsilon_{eig}(p) &= 1- \frac{M_{eig}(p)}{M_{eff}(p)}, \\
        \epsilon_{deg}(p) &= 1- \frac{M_{deg}(p)}{M_{eff}(p)} 
\end{eqnarray}
In general, for any centrality measure $c$, the imprecision function $\epsilon_{c}(p)$ is defined as 
\begin{equation}     
     \epsilon_{c} (p)  = 1- \frac{M_{c}(p)}{M_{eff}(p)}
     \end{equation}
     
\section{Experimental Setup}\label{data}

In this section we describe our experimental setup and the datasets we used.  All simulation and centrality analysis was done in Version 2.14.1 of \verb'R' \cite{R}.  The operating system used was Windows Vista Enterprise (32 bit) and the computer had an Intel Core 2 Quad CPU (Q9650) 3.0 GHz with 4 GB of RAM. Run times to analyze the networks ranged from several hours for the small networks to several days for the larger ones. Centrality measures were computed using the \verb'igraph' \cite{igraph} package in \verb'R'.

We obtained our datasets from a variety of sources.  
Brief descriptions of these networks are as follows:


\begin{itemize}[leftmargin=0cm]
    \item[]{\textbf{cond-mat-GCC} is an academic collaboration network from the e-print arXiv and covers scientific collaborations between authors' papers submitted to Condensed Matter category from 1999~\cite{umich}.}
    \item[]{\textbf{ca-GrQc-GCC} is an academic collaboration network from the e-print arXiv and covers scientific collaborations between authors' papers submitted to the General Relativity and Quantum Cosmology category from Jan. 1993 - Apr. 2003~\cite{snap}.}
    \item[]{\textbf{urv-email} is an e-mail network based on communications of members of the University Rovira i Virgili (Tarragona) \cite{uvi}.  It was extracted in 2003.}
    \item[]{\textbf{1-edges-GCC} is a network formed from YouTube, the video-sharing website that allows users to establish friendship links~\cite{Zafarani+Liu:2009}.  The sample was extracted in Dec. 2008. Links represent two individuals sharing one or more subscriptions to channels on YouTube.}
    \item[]{\textbf{std-GCC} is an online sex community in Brazil in which links represent that one of the individuals posted online about a sexual experience with the other individual, resulting in a bipartite graph. The data was extracted from September of 2002 to October of 2008 \cite{rocha2010}.}
    \item[]{\textbf{as20000102} is a one day snapshot of Internet routers as constructed from the border gateway protocol logs \cite{snap}.  It was extracted on Jan 2nd, 2000.}
    \item[]{\textbf{oregon\_010331} is a network of Internet routers over a one week period as inferred from Oregon route-views, looking glass data, and routing registry from covering the week of March 3rd, 2001~\cite{snap}.}
    \item[]{\textbf{ca-HepTh-GCC} is a collaboration network from the e-print arXiv and covers scientific collaborations between authors' papers submitted to the High Energy Physics - Theory category.  It covers paper from Jan 1993 to Apr 2003 \cite{snap}.}
    \item[]{\textbf{as-22July06} is a snapshot of the Internet on 22 July 2006 at the autonomous systems level compiled by Mark Newman \cite{umich}.}
    \item[]{\textbf{netscience-GCC} is a network of coauthorship of scientists working on network theory and experiments compiled by Mark Newman in May 2006 \cite{umich}.}
\end{itemize}

All datasets used in this paper were obtained from one of four sources: the ASU Social Computing Data Repository~\cite{Zafarani+Liu:2009}, the Stanford Network Analysis Project~\cite{snap}, Mark Newman's data repository at the University of Michigan~\cite{umich}, and Universitat Rovira i Virgili~\cite{uvi}.  All networks considered were symmetric; i.e., if a directed edge from vertex $v$ to $v'$ exists, there is also an edge from vertex $v'$ to $v$.  Summary statistics for these networks can be found in Table \ref{NetStats}.

In the cases where the network had more than one component, we used only the greatest connected component.  We append the suffix ``-GCC'' when referring to those networks.  For example, the cond-mat network had more than one component, so we will use the greatest connected component and refer to this network as ``cond-mat-GCC''.

\begin{table*} 
{\small
\begin{tabular}{llrrrrcrrrr}
\toprule
Name & Type & Nodes & Edges & Density& $\beta'$ & $\lambda$ & $R^2$ & $\langle k \rangle$ &	$\langle k^2 \rangle $ & $K_S$\\
\midrule
\hline
1-edges-GCC & online & 13679 & 76741 & 0.0008 & 2.3 & 1.8 & 0.90 & 11.2 & 502.6 & 25\\
as20000102 & router & 6474 & 12572 & 0.0006 & 0.6 & 1.2 & 0.73 & 3.9 & 640.0 & 12\\
ca-GrQc-GCC & collab &  4158 & 13422 & 0.0016 & 6.3 & 2.0 & 0.88 & 5.5 & 93.2 & 43\\
cond-mat-GCC & collab & 13861 & 44619 & 0.0005 & 8.4 & 2.4 & 0.93 & 5.9 & 75.6 & 17\\
oregon2\_010331 & router & 10900 & 31180 & 0.0005 & 0.5 & 1.2 & 0.79 & 5.7 & 1188.8 & 31\\
std-GCC & std & 15810 & 38540 & 0.0003 & 3.7 & 1.9 & 0.92 & 4.7 & 130.9 & 11\\
urv-email & email & 1133 & 5451 & 0.0085 & 5.7 & 1.5 & 0.84 & 9.6 & 179.8 & 11\\
ca-HepTh-GCC & collab & 8638 & 24806 & 0.0007 & 8.3 & 2.2 & 0.90 & 5.7 & 74.6 &31\\
as-22July2006 & router & 22963 & 48436 & 0.0002 & 0.4 & 1.2 & 0.72 & 4.2 & 1103.0 & 25\\
netscience-GCC & collab & 379 & 914 & 0.0127 & 14.2 & 1.6 & 0.76 & 4.8 & 38.7 & 8 \\
\bottomrule
\end{tabular}
}
\caption{Network Summary Statistics. Note that $\beta'$ is the minimum threshold of infection rate for the epidemic to spread to a significant portion of the network, $\lambda$ is exponent of the power law of the degree distribution, $R^2$ is goodness of fit between the power law and the degree distribution, $\langle k \rangle$ and $\langle k^2 \rangle$ are the first and second moments of the degree distribution, and $K_S$ is the maximum shell present in the network.}
\label{NetStats}
\end{table*}

As seen in the Table 1, all networks used are approximately scale free.  This does not infer that they were generated using a preferential attachment model (as introduced in \cite{barabasi}), as many mechanisms can be responsible for generating scale free networks. If they were generated using a preferential attachment model then we would see a correlation between shell number and degree. This would also mean that degree centrality and shell number would have little difference in predicting spreaders, but our simulations show otherwise.  Figure~\ref{corevsdegree} shows an example in which degree and shell number are not correlated.

\begin{center}
\begin{figure} 
\centering
\includegraphics[width=0.8\linewidth]{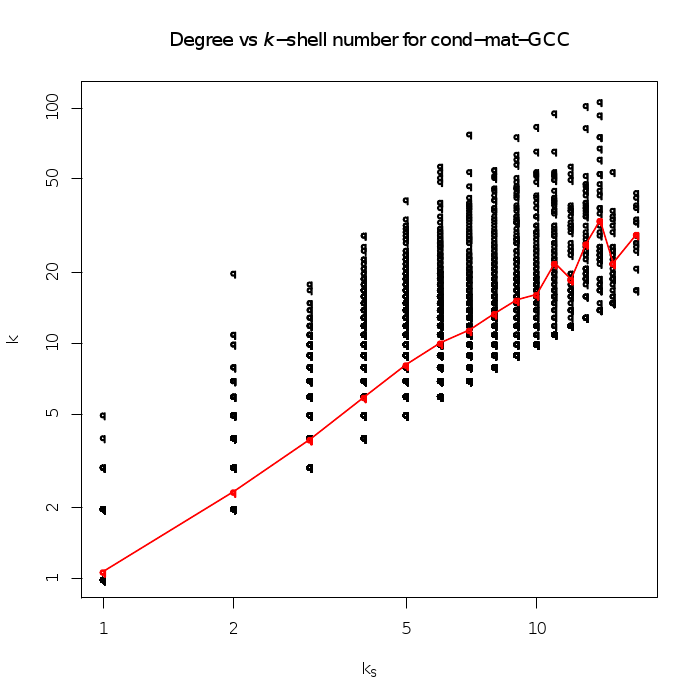} 
\caption{\textit{In the higher shells of these two examples, degree and shell number are not correlated, indicating these can not be assumed to be generated by preferential attachment models. The red line shows the average degree of each shell.  Note that log scales are being used on both axes.}}
\label{corevsdegree}
\end{figure}
\end{center}

\section{Results}\label{results}
Earlier we noted that (1) the relative performance of degree, shell number and other centrality measures can depend on the $\beta$ parameter of the SIR model, and (2) eigenvector centrality performs very well in general regardless of the value of $\beta$ used, typically outperforming all of the other centrality measures that we tried. Here we present more results illustrating these two points. Unless otherwise specified, the $\beta$ values that we used when plotting the imprecision function versus $\beta$ are $1.1\beta', 1.2\beta', \ldots, 2.0\beta',$ where $\beta'$ is the epidemic threshold for the network in question.

\subsection{Sensitivity to $\beta$}
    In Figures \ref{cond-mat-1} and \ref{cond-mat-2}, we saw that the performance of degree relative to shell number changes with $\beta$ for the cond-mat network.  For $\beta=11.17$, shell number was a better indicator of spreading, but for $\beta=15.95$, degree was better. Another way that we could depict this dependence on $\beta$ is to fix $p$ and plot the imprecision versus $\beta$, instead of fixing $\beta$ and plotting the imprecision versus $p$.  In Figure \ref{cond-mat-ineff-vs-beta}, we fix $p=5$ and plot the imprecision function of degree, shell number, and eigenvector centrality versus $\beta$, for $\beta$ between $11.17$ and $15.95$.
        \begin{figure}
        \begin{center}
        \includegraphics[width=0.8\linewidth]{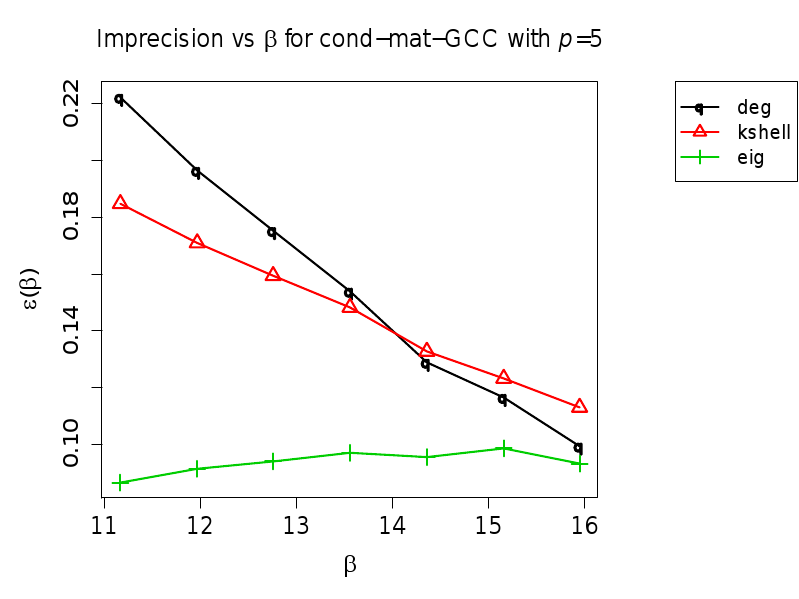}
        \caption{\textit{Imprecision vs $\beta$ for the cond-mat network.  The relative performance of degree and shell number changes near $\beta=14$. }}
        \label{cond-mat-ineff-vs-beta}
        \end{center}
        \end{figure}
    Notice that at around $\beta=14$, degree begins to outperform shell number.
    
    The relative performance of other centrality measures can change as well.  In Figure \ref{ca-GrQc-GCC-ineff-vs-beta}, we plot the imprecision functions of degree, shell number, eigenvector, and closeness centrality versus $\beta$ for $p=5$.  
        \begin{figure}
        \begin{center}
        \includegraphics[width=0.8\linewidth]{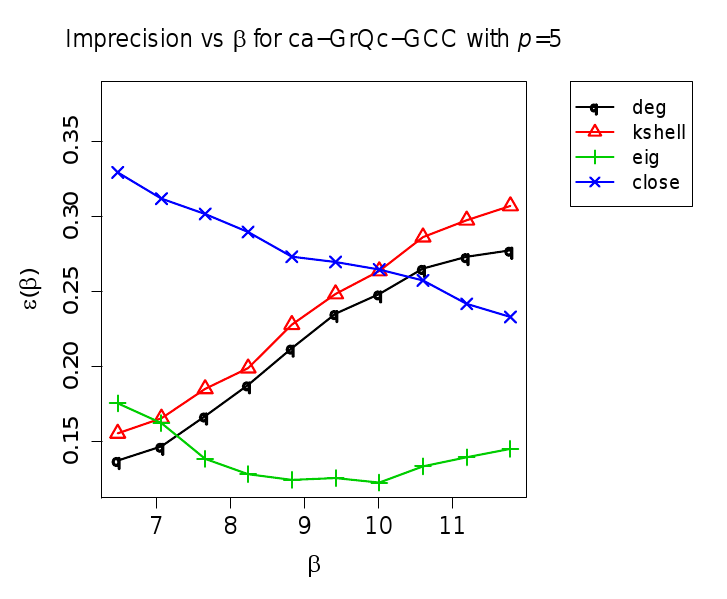}
        \caption{\textit{Imprecision vs $\beta$ for the ca-GrQc-GCC network. }}
        \label{ca-GrQc-GCC-ineff-vs-beta}
        \end{center}
        \end{figure}
    In this network, for $\beta$ near $\beta'$, degree and shell number perform very well.  However, as $\beta$ increases, the imprecision functions of those measures increase, and other measures, like closeness and eigenvector, outperform degree and shell number.

\subsection{Eigenvector centrality}
    As we saw in Figure \ref{eig-kshell}, eigenvector centrality outperforms shell number for all but one of the networks we examined.  Eigenvector centrality also typically outperforms all of the other centrality measures that we tried.  In Figure \ref{fig:imprec-cond-mat-GCC}, we plot the imprecision functions of several different centrality measures for the cond-mat network.  We see that eigenvector centrality performs best for this network. 
        \begin{figure}
        \begin{center}
        \includegraphics[width=.8\linewidth]{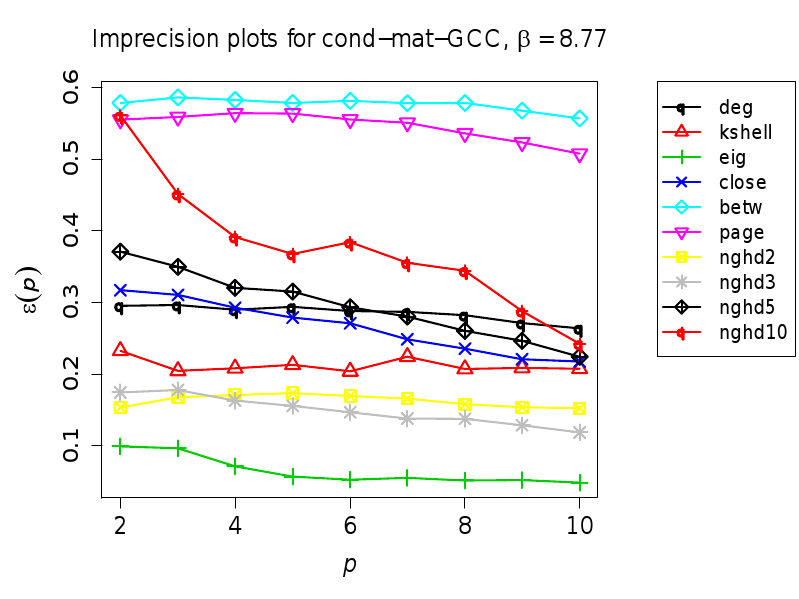}
        \caption{\textit{Imprecision vs $p$ for the cond-mat-GCC network with $\beta=1.1\beta'= 8.77$. We see that eigenvalue centrality performs best for this network.}}
        \label{fig:imprec-cond-mat-GCC}
        \end{center}
        \end{figure}
In Figures \ref{fig:imprec-netscience-GCC},   \ref{fig:imprec-1-edges-GCC}, \ref{fig:imprec-std-GCC}, and \ref{fig:imprec-urv-email}, we give an example of a collaboration network, an online network, an STD network, and an email network in which eigenvector performs best.  

\begin{figure}
\begin{center}
\includegraphics[width=.8\linewidth]{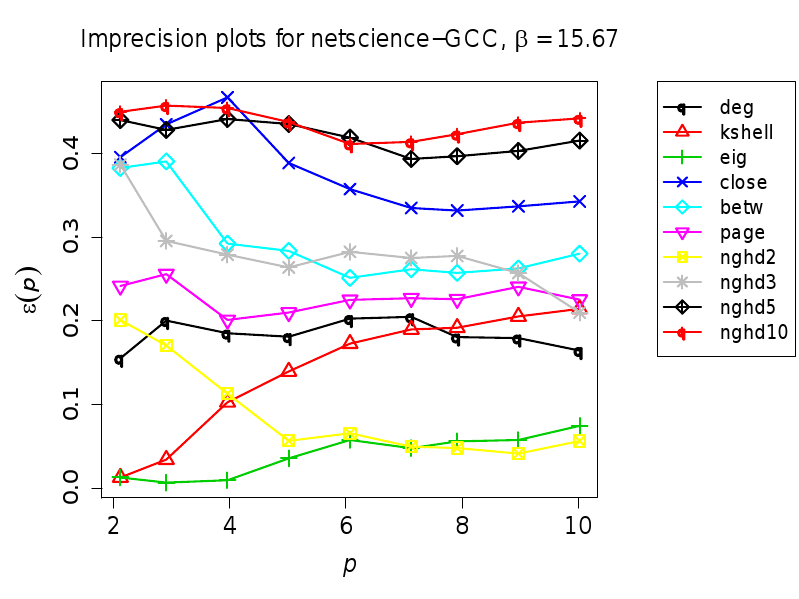}
\caption{\textit{Imprecision vs $p$ for the netscience-GCC network with $\beta=1.1\beta' =15.67$. We see that eigenvalue centrality performs best for this network.}}
\label{fig:imprec-netscience-GCC}
\end{center}
\end{figure}

\begin{figure}
\begin{center}
\includegraphics[width=.8\linewidth]{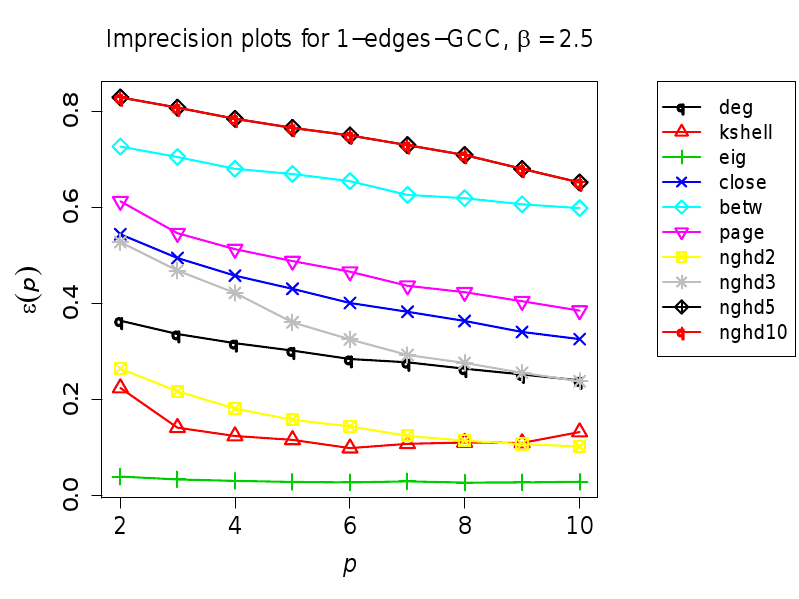}
\caption{\textit{Imprecision vs $p$ for the 1-edges-GCC network with $\beta=1.1\beta' =2.50$. We see that eigenvalue centrality performs best for this network.}}
\label{fig:imprec-1-edges-GCC}
\end{center}
\end{figure}

\begin{figure}
\begin{center}
\includegraphics[width=.8\linewidth]{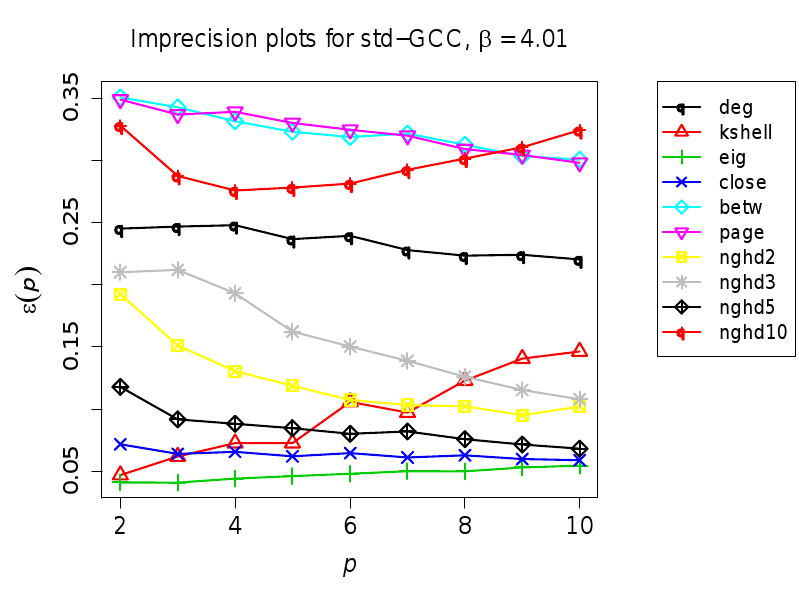}
\caption{\textit{Imprecision vs $p$ for the std-GCC network with $\beta=1.1\beta' =4.01$. We see that eigenvalue centrality performs best for this network.}}
\label{fig:imprec-std-GCC}
\end{center}
\end{figure}

\begin{figure}
\begin{center}
\includegraphics[width=.8\linewidth]{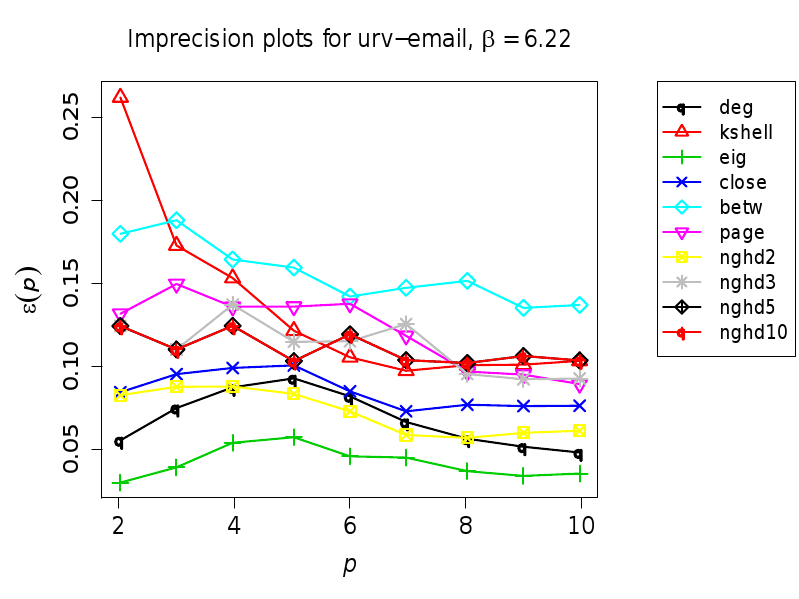}
\caption{\textit{Imprecision vs $p$ for the urv-email network with $\beta=1.1\beta' =6.22$. We see that eigenvalue centrality performs best for this network.}}
\label{fig:imprec-urv-email}
\end{center}
\end{figure}

    Eigenvector centrality did not outperform shell number for the ca-HepTh network, so we can not conclude that eigenvector centrality performs best for \textit{every} network that we tried.  However, it does seem that, \textit{on average}, for the networks we tried, eigenvector centrality performed best for $\beta = 1.1\beta', 1.2 \beta', ..., 2.0\beta'.$  Suppose we take the imprecision functions for $\beta = 1.1 \beta'$ for each network, and we average these imprecision functions over all of our networks, including the ca-HepTh network.   This would be one way to check how well each centrality measure performs on average.  In Figure \ref{avg-ineff-1-1}, we plot this the average imprecision versus $p$ for $\beta = 1.1 \beta'$. We see that, on average, eigenvector centrality outperforms the other measures.  The measure $nghd2$ performed well also.  We give similar figures for $\beta= 1.5\beta'$ and $\beta=2.0\beta'$ in Figures \ref{avg-ineff-1-5} and \ref{avg-ineff-2-0}.  In both cases, eigenvector centrality outperforms all of the other measures.
    
    We believe that eigenvector centrality performs well for some of the same reasons that shell number performs well.  A node has high eigenvector centrality when the node and its neighbors have high degree.  Nghd2, nghd3, and the closely related measure of \cite{chen12} also perform well for this reason.  A hub, or a node with high degree, in the periphery of a network, which does not have many neighbors with high degree, will not typically be as good of a spreader as a node with high eigenvector centrality.

        \begin{figure}
        \begin{center}
        \includegraphics[width=.8\linewidth]{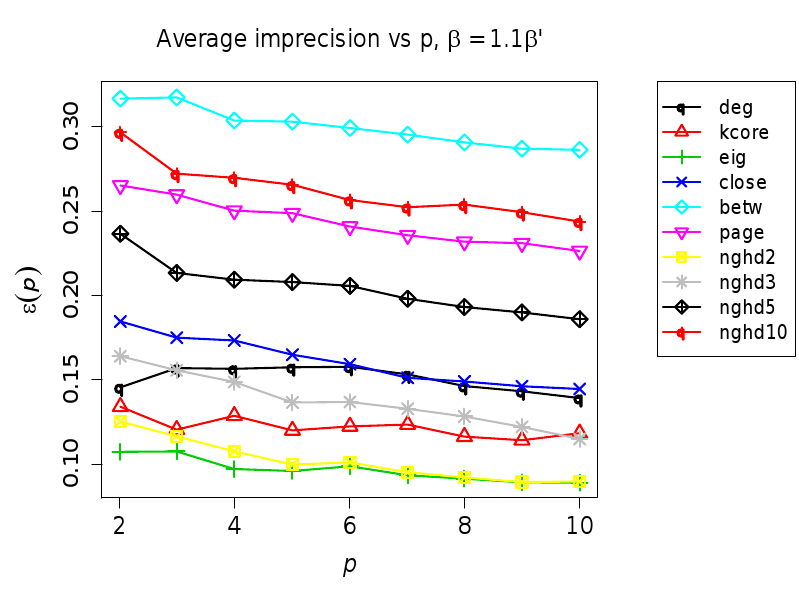}
        \caption{\textit{Average Imprecision vs $p$ with $\beta=1.1\beta'$, where the average is taken over all networks that we considered. }}
        \label{avg-ineff-1-1}
        \end{center}
        \end{figure}

\subsection{Large values of $\beta$}

In \cite{InfluentialSpreaders_2010}, only relatively small values for $\beta$ were explored as it was noted that larger values of $\beta$ would likely cause spreading to a large portion of the population regardless of the location of the initially infected node.  However, in the networks we studied, we found a difference in the ability of the starting node to spread even at seven times the epidemic threshold.  Further, the result that eigenvector centrality performs best, based on average imprecision over all the networks, still holds for these larger values of $\beta$.  We display our imprecision functions for larger values of $\beta$ in Figure~\ref{lg-avg}.  We also show that for five times the epidemic threshold, eigenvector centrality still outperforms the other centrality measures for different values of $p$ (Figure~\ref{lg-ineff-5}).

\begin{figure}
\begin{center}
    \includegraphics[width=.8\linewidth]{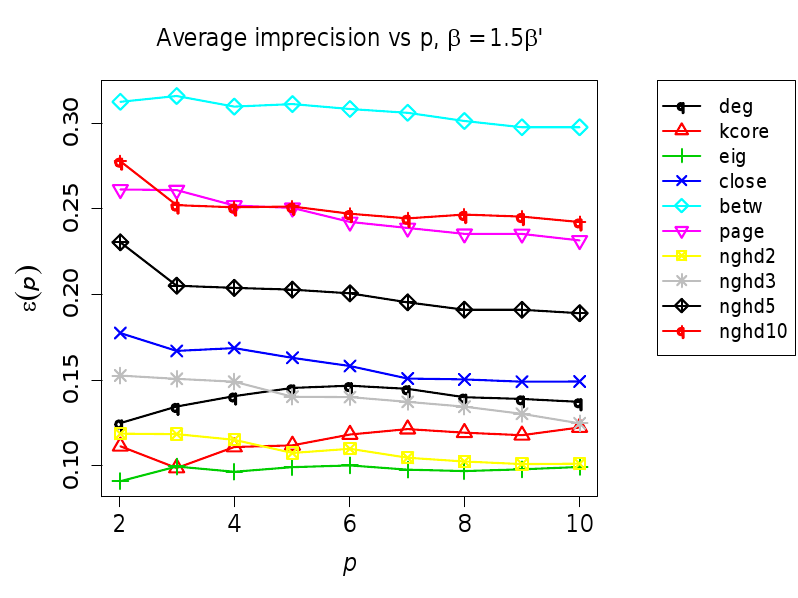}
    \caption{\textit{Average Imprecision vs $p$ with $\beta=1.5\beta'$, where the average is taken over all networks that we considered. We see that, on average, eigenvector performs best.}}
    \label{avg-ineff-1-5}
\end{center}
\end{figure}

\begin{figure}
\begin{center}
        \includegraphics[width=.8\linewidth]{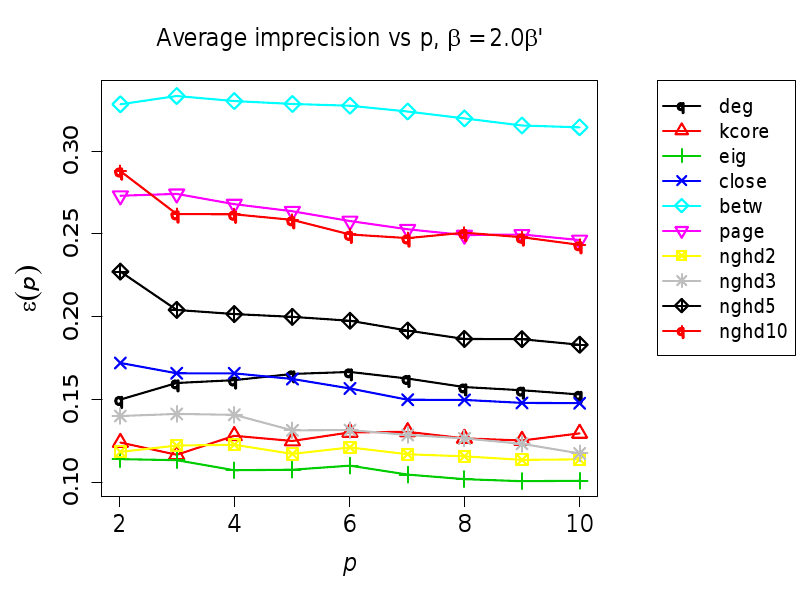}
        \caption{\textit{Average Imprecision vs $p$ with $\beta=2.0\beta'$, where the average is taken over all networks that we considered. We see that, on average, eigenvector performs best.}}
        \label{avg-ineff-2-0}
\end{center}
\end{figure}

\begin{figure}
\begin{center}
    \includegraphics[width=.8\linewidth]{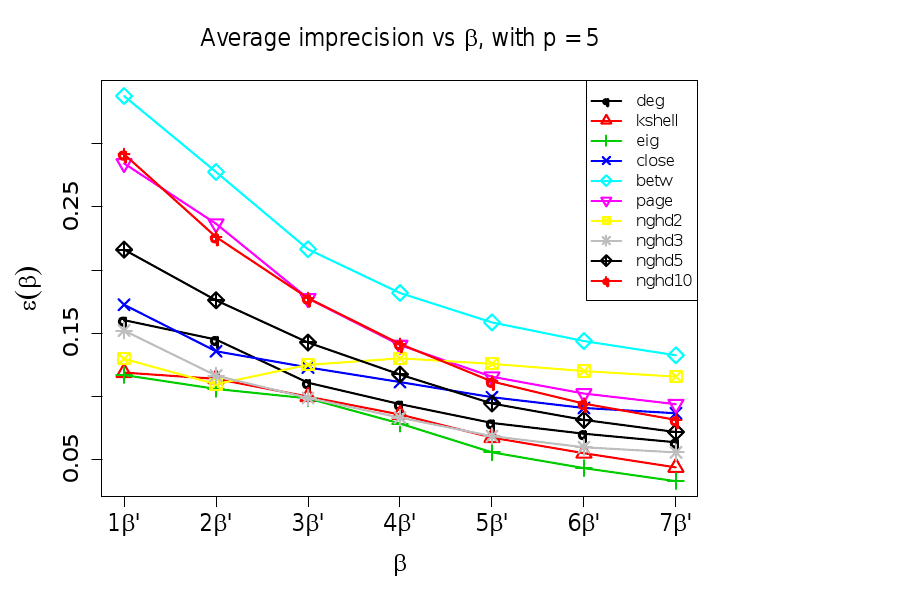}
    \caption{\textit{Average Imprecision vs. $\beta$ with $p=5$.  We see that, on average, eigenvector performs best.}}
    \label{lg-avg}
\end{center}
\end{figure}

\begin{figure}
\begin{center}
        \includegraphics[width=.8\linewidth]{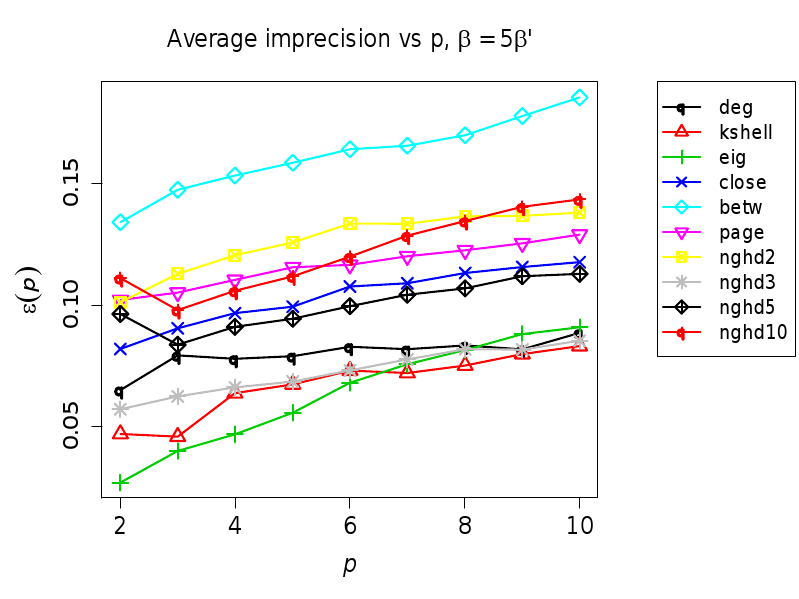}   
    \caption{\textit{Average Imprecision vs. $p$ with $\beta=5\beta'$, where the average is taken over all networks that we considered.  We see that, on average, eigenvector performs best.}}
        \label{lg-ineff-5}
\end{center}
\end{figure}

\FloatBarrier

\section{Conclusions and Future Work}
	These new experiments provide further insight into the issue of identifying spreaders in complex networks that was initiated by \cite{InfluentialSpreaders_2010}.  We extended their work by studying multiple values of the infection probability $\beta$ and showed that the relative ability for centrality measures to identify spreaders often depends on this parameter.  We also noted that eigenvector centrality consistently outperforms the other centrality measures, usually independent of $\beta$.  Future work on identifying influential spreaders could include identifying nodes that not only cause significant spreading, but do so quickly, thus accounting for the time it takes for individuals in the population to become infected.  Further, it would be also interesting to examine which centrality measures best identify spreaders in non-monotonic models of diffusion processes, such as the voter model.  Another aspect for future work would be to examine group centrality. In other words, one could use a centrality measure on sets of nodes to identify the best set of spreaders under the SIR model \cite{spreaders2}.  Finally, it is also worth empirically studying centrality measures designed specifically for the SIR model or other diffusion process, as described in recent work such as \cite{klemm12} and \cite{subrah12}. However, we note that one key advantage to the approach taken in this paper is that the centrality measures studied are already well established - and hence common in many software tools for complex network analysis.
	
	\bigskip
	
\FloatBarrier

\vfill

\section*{Acknowledgments}
Some of the authors of this paper are supported under OSD project F1AF262025G001 and ARO project 2GDATXR042.  The authors are very thankful to these organizations for their support.

We would like to thank Jon Bentley for his feedback on an earlier version of this paper.

The views expressed in this article are those of the authors and do not reflect the official policy or position of the United States Military Academy, the Department of the Army, the Department of Defense, the United States Government, or any of the listed funding agencies.


\begin{thebibliography}{}

\bibitem[\protect\citename{Albert-László~Barabási, }1999]{barabasi}
Albert-László~Barabási, Réka~Albert. (1999).
\newblock Emergence of scaling in random networks.
\newblock {\bf 286}(5439), 509--512.

\bibitem[\protect\citename{Anderson \& May, }1979]{anderson79}
Anderson, Roy~M., \& May, Robert~M. (1979).
\newblock Population biology of infectious diseases: Part i.
\newblock {\em Nature}, {\bf 280}(5721), 361.

\bibitem[\protect\citename{Antal {\em et~al.}\relax, }2006]{antal06}
Antal, T., Redner, S., \& Sood, V. (2006).
\newblock Evolutionary dynamics on degree-heterogeneous graphs.
\newblock {\em Physical review letters}, {\bf 96}(18), 188104.

\bibitem[\protect\citename{Arenas, }2012]{uvi}
Arenas, Alex. (2012).
\newblock {\em Network data sets}.

\bibitem[\protect\citename{Bonacich, }1972]{bona72}
Bonacich, Phillip. (1972).
\newblock Factoring and weighting approaches to status scores and clique
  identification.
\newblock {\em The journal of mathematical sociology}, {\bf 2}(1), 113--120.

\bibitem[\protect\citename{Borgatti \& Everett, }2006]{borgatti06}
Borgatti, S., \& Everett, M. (2006).
\newblock {A Graph-theoretic perspective on centrality}.
\newblock {\em Social networks}, {\bf 28}(4), 466--484.

\bibitem[\protect\citename{Borge-Holthoefer \& Moreno, }2012]{borge12a}
Borge-Holthoefer, Javier, \& Moreno, Yamir. (2012).
\newblock Absence of influential spreaders in rumor dynamics.
\newblock {\em Phys. rev. e}, {\bf 85}(026116).

\bibitem[\protect\citename{Borge-Holthoefer {\em et~al.}\relax, }2012]{borge12}
Borge-Holthoefer, Javier, Rivero, Alejandro, \& Moreno, Yamir. (2012).
\newblock Locating privileged spreaders on an online social network.
\newblock {\em Phys. rev. e}, {\bf 85}(Jun), 066123.

\bibitem[\protect\citename{Brandes, }2001]{brandes01}
Brandes, Ulrik. (2001).
\newblock A faster algorithm for betweenness centrality.
\newblock {\em Journal of mathematical sociology}, {\bf 25}(163).

\bibitem[\protect\citename{Callaway {\em et~al.}\relax, }2000]{callaway00}
Callaway, Duncan~S., Newman, M. E.~J., Strogatz, Steven~H., \& Watts, Duncan~J.
  (2000).
\newblock Network robustness and fragility: Percolation on random graphs.
\newblock {\em Phys. rev. lett.}, {\bf 85}(Dec), 5468--5471.

\bibitem[\protect\citename{Carmi {\em et~al.}\relax, }2007]{ShaiCarmi07032007}
Carmi, Shai, Havlin, Shlomo, Kirkpatrick, Scott, Shavitt, Yuval, \& Shir, Eran.
  (2007).
\newblock {From the Cover: A model of Internet topology using k-shell
  decomposition}.
\newblock {\em Pnas}, {\bf 104}(27), 11150--11154.

\bibitem[\protect\citename{Castellano \& Pastor-Satorras, }2012]{castellano12}
Castellano, \& Pastor-Satorras, Romualdo. (2012).
\newblock Competing activation mechanisms in epidemics on networks.
\newblock {\em Scientific reports}, {\bf 2}(371).

\bibitem[\protect\citename{Chen {\em et~al.}\relax, }2012]{chen12}
Chen, Duanbing, Lü, Linyuan, Shang, Ming-Sheng, Zhang, Yi-Cheng, \& Zhou, Tao.
  (2012).
\newblock Identifying influential nodes in complex networks.
\newblock {\em Physica a: Statistical mechanics and its applications}, {\bf
  391}(4), 1777 -- 1787.

\bibitem[\protect\citename{Chen {\em et~al.}\relax, }2010]{chen10}
Chen, Wei, Wang, Chi, \& Wang, Yajun. (2010).
\newblock Scalable influence maximization for prevalent viral marketing in
  large-scale social networks.
\newblock {\em Pages  1029--1038 of:} {\em Proceedings of the 16th acm sigkdd
  international conference on knowledge discovery and data mining}.
\newblock KDD '10.
\newblock New York, NY, USA: ACM.

\bibitem[\protect\citename{Cohen {\em et~al.}\relax, }2000]{cohen00}
Cohen, Reuven, Erez, Keren, ben Avraham, Daniel, \& Havlin, Shlomo. (2000).
\newblock {Resilience of the Internet to Random Breakdowns}.
\newblock {\em Physical review letters}, {\bf 85}(21), 4626--4628.

\bibitem[\protect\citename{Csardi \& Nepusz, }2006]{igraph}
Csardi, Gabor, \& Nepusz, Tamas. (2006).
\newblock The igraph software package for complex network research.
\newblock {\em Interjournal}, {\bf Complex Systems}, 1695.

\bibitem[\protect\citename{Freeman, }1977]{freeman77}
Freeman, Linton~C. (1977).
\newblock A set of measures of centrality based on betweenness.
\newblock {\em Sociometry}, {\bf 40}(1), pp. 35--41.

\bibitem[\protect\citename{Freeman, }1979]{freeman79cent}
Freeman, Linton~C. (1979).
\newblock Centrality in social networks conceptual clarification.
\newblock {\em Social networks}, {\bf 1}(3), 215 -- 239.

\bibitem[\protect\citename{J.~Goldenberg, }2001]{libai}
J.~Goldenberg, B.~Libai, E.~Muller. (2001).
\newblock Talk of the network: A complex systems look at the underlying process
  of word-of-mouth.
\newblock {\em Marketing letters}, {\bf 12}(3), 211.

\bibitem[\protect\citename{Kang {\em et~al.}\relax, }2012]{subrah12}
Kang, C., Molinaro, C., Kraus, S., Shavitt, Y., \& Subrahmanian, V.S. 2012
  (Aug.).
\newblock Diffusion centrality in social networks.
\newblock  {\em Proc. 2012 ieee/acm intl. conf. on advances in social networks
  analysis and mining (asonam-12)}.

\bibitem[\protect\citename{Kempe {\em et~al.}\relax, }2003]{kleinberg}
Kempe, David, Kleinberg, Jon, \& Tardos, \'{E}va. (2003).
\newblock Maximizing the spread of influence through a social network.
\newblock {\em Pages  137--146 of:} {\em Kdd '03: Proceedings of the ninth acm
  sigkdd international conference on knowledge discovery and data mining}.
\newblock New York, NY, USA: ACM.

\bibitem[\protect\citename{Kitsak {\em et~al.}\relax,
  }2010]{InfluentialSpreaders_2010}
Kitsak, Maksim, Gallos, Lazaros~K., Havlin, Shlomo, Liljeros, Fredrik, Muchnik,
  Lev, Stanley, H.~Eugene, \& Makse, Hernan~A. (2010).
\newblock {Identification of influential spreaders in complex networks}.
\newblock {\em Nat phys}, {\bf 6}(11), 888--893.

\bibitem[\protect\citename{Klemm {\em et~al.}\relax, }2012]{klemm12}
Klemm, Konstantin, Serrano, M.~Angeles, Eguiluz, Victor~M., \& San~Miguel,
  Maxi. (2012).
\newblock {A measure of individual role in collective dynamics: spreading at
  criticality}.
\newblock {\em Scientific reports}, {\bf 2}(292).

\bibitem[\protect\citename{Leskovec, }2012]{snap}
Leskovec, Jure. (2012).
\newblock {\em Stanford network analysis project (snap)}.

\bibitem[\protect\citename{Luis E. C.~Rocha \& Holme, }2010]{rocha2010}
Luis E. C.~Rocha, Fredrik~Liljeros, \& Holme, Petter. (2010).
\newblock Information dynamics shape the sexual networks of internet-mediated
  prostitution.
\newblock {\em Proceedings of the national academy of sciences}, March.

\bibitem[\protect\citename{Madar {\em et~al.}\relax, }2004]{madar04}
Madar, N., Kalisky, T., Cohen, R., ben Avraham, D., \& Havlin, S. (2004).
\newblock {Immunization and epidemic dynamics in complex networks}.
\newblock {\em The european physical journal b - condensed matter and complex
  systems}, {\bf 38}(2), 269--276.

\bibitem[\protect\citename{Moores {\em et~al.}\relax, }2012]{spreaders2}
Moores, Geoffrey, Shakarian, Paulo, Howard, Nicholas, \& Macdonald, Brian.
  (2012).
\newblock {Influential Spreaders 2}.
\newblock \textit{In progress}.

\bibitem[\protect\citename{Newman, }2011]{umich}
Newman, Mark. (2011).
\newblock {\em Network data}.

\bibitem[\protect\citename{Page {\em et~al.}\relax, }1998]{Page98}
Page, L., Brin, S., Motwani, R., \& Winograd, T. (1998).
\newblock The pagerank citation ranking: Bringing order to the web.
\newblock {\em Pages  161--172 of:} {\em Proceedings of the 7th international
  world wide web conference}.

\bibitem[\protect\citename{{R Development Core Team}, }2011]{R}
{R Development Core Team}. (2011).
\newblock {\em R: A language and environment for statistical computing}.
\newblock R Foundation for Statistical Computing, Vienna, Austria.
\newblock {ISBN} 3-900051-07-0.

\bibitem[\protect\citename{Seidman, }1983]{Seidman83}
Seidman, Stephen~B. (1983).
\newblock Network structure and minimum degree.
\newblock {\em Social networks}, {\bf 5}(3), 269 -- 287.

\bibitem[\protect\citename{Valiant, }1979]{valiant79}
Valiant, Leslie~G. (1979).
\newblock The complexity of enumeration and reliability problems.
\newblock {\em Siam j. comput.}, {\bf 8}(3), 410--421.

\bibitem[\protect\citename{Wasserman \& Faust, }1994]{wasserman1994social}
Wasserman, Stanley, \& Faust, Katherine. (1994).
\newblock {\em Social network analysis: Methods and applications}. 1 edn.
\newblock Structural analysis in the social sciences, no. ~8.
\newblock Cambridge University Press.

\bibitem[\protect\citename{Zafarani \& Liu, }2009]{Zafarani+Liu:2009}
Zafarani, R., \& Liu, H. (2009).
\newblock {\em Social computing data repository at {ASU}}.

\end{thebibliography}

\end{document}